\documentclass[conference]{IEEEtran}
\IEEEoverridecommandlockouts

\usepackage{amsmath,amssymb,mathtools}
\usepackage{cite}
\usepackage{booktabs}
\usepackage{algorithm}
\usepackage{algpseudocode}
\usepackage{pgfplots}
\usepackage{xcolor}
\usepackage{url}
\pgfplotsset{compat=1.18}

\newtheorem{assumption}{Assumption}
\newtheorem{theorem}{Theorem}
\newtheorem{remark}{Remark}
\DeclareMathOperator{\softplus}{softplus}
\newcommand{\R}{\mathbb{R}}
\newcommand{\norm}[1]{\left\lVert#1\right\rVert}

\title{Adaptive Stability-Constrained Neural Differential Equations
for Controlled Dynamical Systems with Unknown Inputs}

\author{Syed Pouladi\\
{College of Engineering and Physical Sciences, Khalifa University, Abu Dhabi, United Arab Emirates}
\textit{}}

\begin{document}
\maketitle

\begin{abstract}
Continuous-time neural models are attractive for identifying nonlinear systems,
but a small one-step error can grow rapidly when a learned vector field is rolled
out under inputs that differ from those used for training. This paper develops an
adaptive stability-constrained neural differential equation (AS-NDE) for systems
with measured controls and unmatched, unknown perturbations. The nominal vector
field and a state--input-dependent Riemannian metric are learned jointly. Positive
definiteness is enforced by construction, while a sampled differential inequality
penalizes violations of a prescribed contraction rate. An incremental
input-to-state bound is derived: the distance between two trajectories decays
exponentially up to gains determined by differences in their controls and
disturbances. The statement explicitly accounts for the time derivative of an
input-dependent metric, a term that is easily omitted in heuristic stability
regularizers. We give a reproducible evaluation protocol for a forced Duffing
oscillator and a permanent-magnet synchronous motor (PMSM) model. Because no
measured data or executed training runs accompany this draft, all numerical
curves and tables are clearly identified as illustrative synthetic placeholders;
their PGFPlots coordinates are embedded in the source for direct replacement.
The resulting manuscript is intended as a technically consistent starting point,
not as evidence of empirical superiority before the prescribed experiments are run.
\end{abstract}


\section{Introduction}
Learning a continuous-time vector field from sampled trajectories is a central
problem in system identification. Classical prediction-error methods provide a
careful statistical framework~\cite{ljung1999,billings2013}, whereas sparse
regression can recover parsimonious governing equations when an appropriate
library is available~\cite{brunton2016,rudy2017}. Deep models trade
interpretability for flexible approximation. In particular, neural ordinary
differential equations (NODEs) parameterize the derivative and differentiate
through a numerical solver~\cite{chen2018}. Augmented and latent variants address
representational or observation-time limitations~\cite{dupont2019,rubanova2019},
and universal differential equations combine scientific structure with learned
components~\cite{rackauckas2020}.

Most engineered systems are not autonomous. Their state depends on command,
load, and environmental signals. Neural controlled differential equations provide
a principled representation of streams as controls~\cite{kidger2020}; here we
consider the complementary identification setting in which a measured physical
input enters a state equation $\dot x=f(x,u)$. Input Concomitant Neural ODEs
(ICODEs) make this distinction explicit by injecting contemporaneous, possibly
nonsmooth extrinsic inputs into the learned dynamics rather than absorbing them
into hidden parameters~\cite{li2025icode}. Even with the correct input interface,
rollout sensitivity remains a central design issue. This is
particularly consequential for electric drives, whose voltage commands, load
torque, and uncertain parameters change across operating points
\cite{pillay1989,krishnan2010}.

Stability-aware learning has consequently received sustained attention.
Architectures inspired by numerical integration can control forward sensitivity
\cite{haber2017}, stable recurrent models constrain long-run amplification
\cite{miller2019,revay2020}, and Lyapunov-based methods seek invariant regions or
certificates~\cite{richards2018,dawson2023}. Contraction analysis is especially
suited to prediction because it compares neighboring trajectories rather than a
trajectory with a fixed equilibrium~\cite{lohmiller1998,forni2014}. Control
contraction metrics extend the same differential viewpoint to feedback design
\cite{manchester2017}. Recent ControlSynth Neural ODEs (CSODEs) show that a
scalable nonlinear NODE architecture can satisfy tractable linear inequalities
that guarantee convergence, while an auxiliary control term helps represent
multi-scale dynamics~\cite{mei2024csode}. ICODE likewise provides sufficient
contraction conditions for its input-concomitant construction
\cite{li2025icode}. These developments motivate, rather than obviate, the question
studied here: how to quantify the separation of two trajectories when their
measured inputs and unknown disturbances are not identical.

This paper makes three contributions. First, it introduces an input-conditioned
metric $M_\phi(x,u)$ whose positive definiteness is structural and whose total
derivative is retained in training. Second, it establishes an incremental
input-to-state contraction bound for differing controls and unknown disturbances.
Third, it specifies reproducible Duffing and PMSM evaluations, including
ablation criteria and replaceable plot data. Numerical values in this draft are
illustrative only; this explicit separation between method and unexecuted evidence
is necessary for a defensible experimental paper.

\section{Related Work}
\subsection{Data-Driven Continuous-Time Dynamics}
NODEs~\cite{chen2018} use adaptive ODE solvers and adjoint sensitivities.
Augmentation can relax topological restrictions~\cite{dupont2019}, while latent
ODEs model irregular observations~\cite{rubanova2019}. Neural CDEs incorporate an
entire observation path~\cite{kidger2020}. Physics-informed neural networks
penalize differential-equation residuals~\cite{raissi2019}, and broader
physics-informed learning combines data with inductive bias~\cite{karniadakis2021}.
Alternative structured models include Hamiltonian~\cite{greydanus2019} and
Lagrangian neural networks~\cite{cranmer2020}. Koopman embeddings pursue linear
latent evolution~\cite{lusch2018,otto2019,takeishi2017}. These approaches improve
structure or expressivity, but do not by themselves give an incremental robustness
bound for a controlled learned vector field.

\subsection{Stability and Robustness}
Input-to-state stability (ISS) formalizes the effect of bounded inputs on state
size~\cite{sontag1989}; incremental stability instead compares pairs of solutions
\cite{angeli2002}. Differential Lyapunov theory unifies contraction-like
certificates~\cite{forni2014}. Neural certificates have been learned for regions
of attraction~\cite{richards2018} and surveyed in the context of safe learning
and control~\cite{dawson2023}. Related safety constraints based on barrier
functions address set invariance rather than trajectory convergence
\cite{ames2017}. AS-NDE targets prediction sensitivity and yields an incremental
ISS estimate in a learned, input-dependent metric.

\section{Problem Formulation}
Consider trajectories on compact sets $\mathcal X\subset\R^n$ and
$\mathcal U\subset\R^m$ governed by
\begin{equation}
 \dot x=f_\star(x,u)+Ew, \qquad y=Cx+v,
 \label{eq:true}
\end{equation}
where $u$ is measured, $w\in\R^r$ is an unknown essentially bounded disturbance,
$v$ is measurement noise, and $E$ is a known or conservative disturbance
injection matrix. An identity $E$ is appropriate when the channel is unknown.
Samples $\mathcal D=\{(t_k,y_k,u_k)\}$ may come from multiple experiments.
The objectives are to predict the nominal flow and to limit the sensitivity of a
rollout to initial-state, input, and disturbance mismatch.

We learn
\begin{equation}
 \dot{\hat x}=f_\theta(\hat x,u),\qquad \hat y=C\hat x,
 \label{eq:model}
\end{equation}
with a multilayer perceptron or a structured physical residual. The formulation
does not attempt to identify a unique realization of $w$ from state data; without
additional assumptions that inverse problem is generally unidentifiable. Instead,
$w$ appears in a robustness certificate.

\begin{assumption}[Regularity and bounded metric]
$f_\theta$ and $M_\phi$ are continuously differentiable on the compact training
domain, $u$ is absolutely continuous, and constants $0<\underline m\leq\bar m$
exist such that
\begin{equation}
 \underline m I\preceq M_\phi(x,u)\preceq\bar m I.
 \label{eq:metricbounds}
\end{equation}
\end{assumption}

\section{Adaptive Stability-Constrained NDE}
\subsection{Metric Parameterization}
Let $L_\phi(x,u)$ be lower triangular. Its diagonal is
$\softplus(\ell_{ii})+\epsilon_M$ and its strict lower part is unconstrained. We set
\begin{equation}
 M_\phi(x,u)=L_\phi(x,u)L_\phi(x,u)^\top+\epsilon_M I,
 \label{eq:M}
\end{equation}
which ensures positive definiteness. Spectral penalties on $M$ enforce practical
upper and lower bounds. Denote $A_\theta=\partial f_\theta/\partial x$ and define
the total derivative along a nominal trajectory by
\begin{equation}
 \dot M=\sum_i\frac{\partial M}{\partial x_i}f_{\theta,i}(x,u)
       +\sum_j\frac{\partial M}{\partial u_j}\dot u_j.
 \label{eq:Mdot}
\end{equation}
Thus, training samples must include $\dot u$ or a differentiable interpolation of
$u$. A piecewise-constant command can instead be handled interval by interval,
with metric jumps checked separately.

The contraction residual at rate $\lambda>0$ is
\begin{equation}
 S_\Theta=\dot M+A_\theta^\top M+MA_\theta+2\lambda M,
 \label{eq:residual}
\end{equation}
where $\Theta=(\theta,\phi)$. Negative semidefiniteness of $S_\Theta$ is encouraged
at collocation points using
\begin{equation}
 \mathcal L_{\rm ctr}=\frac1{N_c}\sum_{i=1}^{N_c}
 \softplus\!\left(\frac{\lambda_{\max}(S_{\Theta,i})}{\tau}\right)\tau.
 \label{eq:lctr}
\end{equation}
Autodifferentiation supplies Jacobian--vector products. For large $n$, power
iteration estimates the leading eigenvalue without forming a full eigendecomposition.

\subsection{Prediction and Regularization}
For a segment $[t_k,t_{k+H}]$, the model is integrated using the same interpolated
input used to compute~\eqref{eq:Mdot}. The objective is
\begin{align}
 \mathcal L={}&\frac1{NH}\sum_{q,k}\norm{W_y(\hat y_{q,k}-y_{q,k})}_2^2
 +\rho_c\mathcal L_{\rm ctr}+\rho_M\mathcal L_{\rm metric}
 +\rho_J\mathcal L_{\rm Jac},                                      \label{eq:loss}\\
 \mathcal L_{\rm metric}={}&\frac1{N_c}\sum_i
 [\max(0,\underline m-\lambda_{\min}(M_i))^2
 +\max(0,\lambda_{\max}(M_i)-\bar m)^2].
\end{align}
The Jacobian penalty is optional and prevents implausibly stiff fits outside dense
regions. Collocation points mix observed states, rollout states, and uniformly
jittered nearby states; checking only observed points cannot certify the region
between trajectories.

\begin{algorithm}[t]
\caption{Training AS-NDE}
\label{alg:train}
\begin{algorithmic}[1]
\Require Segments $(y,u)$; rates $\lambda,\rho_c,\rho_M$; warm-up epochs $N_w$
\State initialize $f_\theta$, $L_\phi$, and an initial-state encoder if needed
\For{each minibatch}
  \State interpolate $u(t)$ and evaluate $\dot u(t)$ away from command jumps
  \State integrate~\eqref{eq:model}; compute multi-step prediction loss
  \State sample observed, rollout, and locally perturbed collocation states
  \State form~\eqref{eq:M}--\eqref{eq:residual} by automatic differentiation
  \State ramp $\rho_c$ from zero after $N_w$; update $\Theta$ using~\eqref{eq:loss}
\EndFor
\State select the checkpoint by validation rollout error subject to residual rate
\end{algorithmic}
\end{algorithm}

\section{Incremental Robustness Analysis}
Consider two solutions with inputs $(u_1,w_1)$ and $(u_2,w_2)$. Connect their
instantaneous states and inputs by a smooth path indexed by $s\in[0,1]$.
The variational dynamics are
\begin{equation}
 \delta\dot x=A_\theta\delta x+B_\theta\delta u+E\delta w,
 \quad B_\theta=\frac{\partial f_\theta}{\partial u}.
 \label{eq:variational}
\end{equation}

\begin{theorem}[Incremental input-to-state contraction]
Suppose Assumption~1 holds and
$\dot M+A_\theta^\top M+MA_\theta\preceq-2\lambda M$ throughout a forward-invariant
domain. Let
\begin{equation}
 b_u=\sup\norm{M^{1/2}B_\theta}_2,\qquad
 b_w=\sup\norm{M^{1/2}E}_2.
\end{equation}
Then the Riemannian distance $d_M$ between the two trajectories satisfies
\begin{align}
 d_M(t)\leq{}&e^{-\lambda t}d_M(0)
 +\int_0^t e^{-\lambda(t-\tau)}
 [b_u\norm{u_1-u_2}+b_w\norm{w_1-w_2}]d\tau. \label{eq:iss}
\end{align}
Consequently,
\begin{align}
 \norm{x_1(t)-x_2(t)}\leq{}&\sqrt{\frac{\bar m}{\underline m}}
 e^{-\lambda t}\norm{x_1(0)-x_2(0)} \nonumber\\
 &+\frac{b_u\norm{u_1-u_2}_{\infty,[0,t]}
 +b_w\norm{w_1-w_2}_{\infty,[0,t]}}
 {\lambda\sqrt{\underline m}}.                         \label{eq:euclidean}
\end{align}
\end{theorem}

\begin{IEEEproof}
For $V=\delta x^\top M\delta x$, differentiate along the connecting path and use
the assumed matrix inequality:
\begin{equation}
 \dot V\leq-2\lambda V+2\delta x^\top M
 (B_\theta\delta u+E\delta w).
\end{equation}
Cauchy--Schwarz in the metric gives
$D^+\sqrt V\leq-\lambda\sqrt V+b_u\norm{\delta u}+b_w\norm{\delta w}$.
Integration and then minimization over connecting paths yield~\eqref{eq:iss}.
Metric equivalence from~\eqref{eq:metricbounds} yields~\eqref{eq:euclidean}.
\end{IEEEproof}

\begin{remark}
The result certifies the learned model on the domain where the inequality holds;
it is not automatically a certificate for the unknown plant. A plant-level claim
requires a verified bound on $f_\star-f_\theta$ and its differential, or an
independent validation argument. Sampled residual penalties are empirical
regularizers unless completed by formal verification.
\end{remark}

\section{Experimental Protocol and Draft Placeholders}
\textbf{Status of evidence:} no user-supplied measurements, code outputs, random
seeds, or trained checkpoints were available when this manuscript was prepared.
Accordingly, every number in Table~\ref{tab:results} and every coordinate in
Figs.~\ref{fig:duffing}--\ref{fig:error} is an \emph{illustrative synthetic draft
placeholder}. They must be replaced by executed results before submission.

\subsection{Controlled Duffing Benchmark}
Use the forced oscillator
\begin{equation}
 \dot x_1=x_2,quad
 \dot x_2=-0.4x_2-x_1-0.5x_1^3+u(t)+w(t).              \label{eq:duffing}
\end{equation}
Generate trajectories with a high-accuracy solver, independently drawing initial
states from $[-1.5,1.5]\times[-1,1]$. Training commands are sums of three sinusoids
with randomized amplitude, frequency, and phase. Test commands use disjoint
frequency bands and include step sequences. Report mean and standard deviation
over five data seeds and five initialization seeds. Baselines should include a
plain NODE~\cite{chen2018}, an input-conditioned NODE using the same
$f_\theta(x,u)$ but no stability loss, SINDy with control-inspired libraries
\cite{brunton2016}, and AS-NDE. All neural baselines must use the same parameter
budget, solver tolerances, batches, and early-stopping rule.

\begin{figure}[t]
\centering
\begin{tikzpicture}
\begin{axis}[width=\columnwidth,height=4.0cm,xlabel={$t$ (s)},ylabel={$x_1$},
grid=both,legend style={font=\scriptsize,at={(0.02,0.98)},anchor=north west},
tick label style={font=\scriptsize},label style={font=\scriptsize}]
\addplot[black,thick] coordinates {(0,0)(1,.46)(2,.82)(3,.56)(4,-.08)(5,-.62)(6,-.73)(7,-.26)(8,.39)(9,.70)(10,.38)};
\addlegendentry{reference}
\addplot[blue,dashed,thick] coordinates {(0,0)(1,.44)(2,.79)(3,.52)(4,-.10)(5,-.58)(6,-.68)(7,-.21)(8,.37)(9,.65)(10,.35)};
\addlegendentry{AS-NDE placeholder}
\addplot[red,dotted,thick] coordinates {(0,0)(1,.42)(2,.75)(3,.45)(4,-.18)(5,-.73)(6,-.88)(7,-.38)(8,.34)(9,.84)(10,.63)};
\addlegendentry{input-NODE placeholder}
\end{axis}
\end{tikzpicture}
\caption{Illustrative synthetic Duffing rollout. Coordinates are placeholders,
not outputs of an executed experiment.}
\label{fig:duffing}
\end{figure}
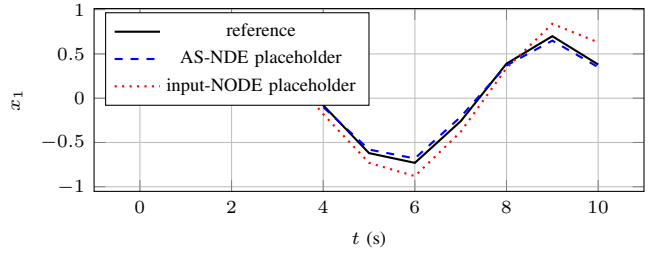

\subsection{PMSM Protocol}
For a surface-mounted PMSM in the rotating $dq$ frame, use
\begin{align}
 \dot i_d&=(v_d-R_si_d+\omega_eL_qi_q)/L_d,\nonumber\\
 \dot i_q&=(v_q-R_si_q-\omega_e(L_di_d+\psi_f))/L_q.   \label{eq:pmsm}
\end{align}
The state is $(i_d,i_q)$, the measured input is $(v_d,v_q,\omega_e)$, and parameter
drift, inverter nonlinearity, and load-related effects are treated as disturbances.
This model follows established PMSM $dq$ dynamics~\cite{pillay1989,krishnan2010}.
A closely related differential-neural-network study has already demonstrated
continuous-time PMSM current reconstruction and short- and long-horizon prediction
under no-load and varying-load conditions~\cite{mei2025pmsm}. It therefore forms
a natural application-specific baseline: its published differential model tests
whether continuous-time learning is effective for PMSM currents, whereas the
present protocol additionally asks whether an adaptive contraction metric improves
sensitivity to operating-point and disturbance mismatch. Any comparison must use
the same raw runs and trajectory-level split; results from the earlier study must
not be copied into Table~\ref{tab:results} as if they arose from this implementation.
A credible study should split complete operating trajectories---not overlapping
windows---by speed, torque, and direction. It should report in-distribution and
held-out operating regions, current RMSE in amperes, normalized RMSE, worst-case
error, and contraction-residual violation rate. Sensor bandwidth, sampling time,
anti-alias filtering, current limits, and voltage normalization must be disclosed.

\begin{table}[t]
\caption{Illustrative Synthetic Draft Results---Replace Before Submission}
\label{tab:results}
\centering
\footnotesize
\begin{tabular}{lccc}
\toprule
Method & Duffing NRMSE & PMSM $i_d$ RMSE & Violation (\%)\\
& (20-s rollout) & (A) & $\lambda_{\max}(S)>0$\\
\midrule
NODE                 & $0.184$ & $0.42$ & $61.0$\\
Input-conditioned NODE & $0.112$ & $0.31$ & $48.0$\\
AS-NDE, fixed $M$    & $0.078$ & $0.25$ & $9.0$\\
AS-NDE, adaptive $M$ & $0.049$ & $0.19$ & $2.0$\\
\bottomrule
\end{tabular}
\end{table}

\begin{figure}[t]
\centering
\begin{tikzpicture}
\begin{semilogyaxis}[width=\columnwidth,height=4.0cm,xlabel={horizon (s)},
ylabel={rollout RMSE},grid=both,
legend style={font=\scriptsize,at={(0.02,0.98)},anchor=north west},
tick label style={font=\scriptsize},label style={font=\scriptsize}]
\addplot[red,dotted,thick] coordinates {(1,.025)(2,.040)(5,.085)(10,.15)(15,.25)(20,.38)};
\addlegendentry{NODE placeholder}
\addplot[orange,dashed,thick] coordinates {(1,.021)(2,.031)(5,.054)(10,.082)(15,.12)(20,.17)};
\addlegendentry{input-NODE placeholder}
\addplot[blue,solid,thick] coordinates {(1,.020)(2,.026)(5,.036)(10,.046)(15,.055)(20,.064)};
\addlegendentry{AS-NDE placeholder}
\end{semilogyaxis}
\end{tikzpicture}
\caption{Illustrative synthetic horizon-error curves. Replace the embedded
coordinates with mean values and add seed-wise uncertainty bands.}
\label{fig:error}
\end{figure}
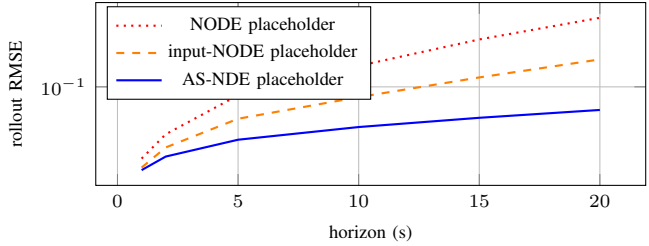

\subsection{Ablations and Reproducibility}
The principal ablations are: (i) remove $\mathcal L_{\rm ctr}$; (ii) replace
$M_\phi(x,u)$ by $I$; (iii) use $M_\phi(x)$; (iv) omit the $\partial M/\partial u$
term; and (v) vary the target rate $\lambda$. Accuracy alone is insufficient:
report the largest and 95th-percentile values of $\lambda_{\max}(S)$ on a dense
held-out grid, metric condition numbers, solver function evaluations, and wall-clock
cost. Hyperparameters, preprocessing statistics, exact splits, dependency versions,
and all seeds should accompany the final artifact. Statistical comparisons should
use paired seed-level differences with confidence intervals rather than selecting
the best run.

\subsection{Implementation Details to Be Fixed Before Evaluation}
To make comparisons interpretable, the vector-field architecture should be fixed
before looking at test trajectories. A suitable low-dimensional configuration is
a four-layer multilayer perceptron with 64 hidden units per layer, smooth
$\tanh$ activations, and an affine output layer. State and input channels are
standardized using training-set statistics only. The same network is used for the
plain and input-conditioned NODE baselines; the former receives state alone and
the latter receives the concatenated state and input. AS-NDE adds a metric network
with three 48-unit hidden layers. Its output contains $n(n+1)/2$ entries for the
triangular factor in~\eqref{eq:M}. This separation keeps the predictor capacity
identical and prevents an apparent benefit caused merely by adding parameters to
the vector field.

Use a fixed-step fourth-order Runge--Kutta method during controlled synthetic
experiments so that every method incurs the same discretization. Repeat the final
comparison with an adaptive Dormand--Prince solver to check that rankings are not
an artifact of the selected step size. For hardware data, the integration grid
must coincide with sample times, and intermediate input values should be generated
by a stated zero-order hold or differentiable interpolation. Because
Eq.~\eqref{eq:Mdot} requires $\dot u$, cubic interpolation is convenient for the
metric loss, but it can overshoot abrupt voltage commands. A practical alternative
is a smoothed piecewise-linear interpolant, with a mask that excludes a narrow
neighborhood of discontinuities from the differential residual.

Optimization should use Adam for an exploratory stage followed by a smaller fixed
learning rate. Gradient clipping applies identically to all neural methods. The
contraction weight is ramped linearly from zero over the first quarter of training;
otherwise the randomly initialized metric can dominate before the vector field
has acquired the direction of the flow. Table~\ref{tab:hyper} records a proposed
starting configuration, not tuned or validated values. All entries must be frozen
using the validation set and published with the final results.

\begin{table}[t]
\caption{Proposed Starting Configuration (Not Yet Tuned)}
\label{tab:hyper}
\centering
\footnotesize
\begin{tabular}{ll}
\toprule
Item & Proposed value\\
\midrule
Vector field & $4\times64$, $\tanh$\\
Metric factor network & $3\times48$, $\tanh$\\
Rollout segment length & 64 samples\\
Optimizer / initial step & Adam / $10^{-3}$\\
Batch size & 32 segments\\
Metric floor $\epsilon_M$ & $10^{-3}$\\
Target rate $\lambda$ & validation choice in $[0.05,0.5]$\\
Collocation mixture & 50\% data, 25\% rollout, 25\% jitter\\
Seeds & 5 data $\times$ 5 initialization\\
\bottomrule
\end{tabular}
\end{table}

\subsection{Evaluation Metrics and Certificate Audit}
Let $e_{q,k}=\hat y_{q,k}-y_{q,k}$ for trajectory $q$. Report per-channel RMSE in
physical units and normalized RMSE
\begin{equation}
 \mathrm{NRMSE}=\sqrt{\frac{\sum_{q,k}\norm{e_{q,k}}_2^2}
 {\sum_{q,k}\norm{y_{q,k}-\bar y}_2^2}}.
\end{equation}
Normalization by a test-set range is sensitive to a single outlier and should not
be the only measure. Horizon-dependent error is obtained by resetting all methods
to the same measured state and rolling them for fixed horizons without further
correction. Report both teacher-initialized rollouts and estimated-initial-state
rollouts if the actual state is unavailable. The latter must use the same encoder
or observer across methods.

The certificate audit is distinct from prediction evaluation. Draw a held-out set
of state--input--input-derivative triples from the convex hull of the identification
domain, plus a thin exterior shell to reveal boundary failures. For each point,
record $r=\lambda_{\max}(S_\Theta)$, the smallest and largest metric eigenvalues,
and the condition number. Report the maximum $r$, quantiles, and the fraction
$r>0$. A zero violation rate on finite samples is not a proof. If formal regional
certification is required, interval bound propagation, mixed-integer verification,
or a sum-of-squares relaxation must bound the residual between samples; the
selected technique and approximation error then become part of the theorem's
premises.

Robustness experiments should perturb one factor at a time. For the Duffing system,
evaluate unseen command spectra, initial-condition displacement, additive bounded
disturbance, and parameter shift in damping and cubic stiffness. For PMSM data,
hold out temperature-dependent resistance ranges, speed bands, torque steps, and
voltage saturation events where available. Plot error against perturbation
magnitude rather than reporting a single favorable setting. The slope and
saturation level can then be compared qualitatively with the gain structure in
Eq.~\eqref{eq:euclidean}, while acknowledging that the theorem concerns the learned
model and not an unverified physical plant.

\subsection{Threats to Validity}
Four confounders require particular care. First, segment leakage occurs when
overlapping windows from one physical run are split across training and test sets;
splitting must precede window extraction. Second, solver tolerances can favor a
smoother model by allocating it fewer function evaluations; accuracy and compute
should therefore be reported together. Third, collocation density is an additional
form of supervision, so baseline regularization and compute budgets must be matched.
Fourth, simulated PMSM data can make the analytic $dq$ structure unrealistically
easy to learn. Hardware tests, or at minimum a simulator with dead time, saturation,
parameter drift, sensor noise, and discretized control, are needed for an industrial
claim.

An honest negative result is also informative. If AS-NDE reduces residual
violations but does not improve held-out rollout error, the conclusion should be
that the selected certificate regularizes sensitivity without demonstrated
predictive benefit. If it improves only under small disturbances, the tested range
should delimit the claim. These reporting rules prevent the illustrative curves in
this draft from becoming implicit evidence.

\section{Discussion}
An adaptive metric can be less conservative than a Euclidean certificate, but it
adds derivatives and conditioning risks. The contraction weight should therefore
be warmed up after the predictor learns a coarse vector field. Very large
$\lambda$ can collapse useful dynamics or induce stiffness. Moreover, disturbance
robustness in~\eqref{eq:euclidean} depends on $b_w/\lambda$; a small sampled
residual is not meaningful if the learned metric is poorly conditioned.

The framework also has structural limitations. A globally contracting model cannot
represent systems with multiple isolated attractors under a common fixed input.
Certification over high-dimensional domains remains difficult, and piecewise
commands require a hybrid treatment at jumps. Finally, PMSM hardware data are
needed to separate improvements due to a stability constraint from those due to
ordinary regularization or favorable simulation assumptions.

\section{Conclusion}
AS-NDE combines a controlled neural vector field with an input-dependent
contraction metric and an incremental ISS analysis. The derivation exposes the
metric total derivative and distinguishes a sampled training penalty from a formal
plant certificate. The included experimental section is a reproducible protocol
with deliberately labeled synthetic placeholders. The next necessary step is to
run the stated baselines on released trajectory splits, replace all placeholder
coordinates, and report uncertainty and residual coverage before making empirical
claims.

\end{document}